\documentclass[11pt]{article}
\usepackage[margin=1in]{geometry} 
\usepackage{amsmath}
\usepackage{tcolorbox}
\usepackage{amssymb}
\usepackage{amsthm}
\usepackage{mathtools}
\usepackage{lastpage}
\usepackage{subfigure}
\usepackage{graphicx}
\usepackage{accents}
\usepackage{systeme}
\usepackage{tikz}
\usepackage{enumerate}
\usepackage{tikz-cd}
\usepackage{enumitem}
\usepackage{bbm}
\usepackage{hyperref}
\usepackage[capitalise]{cleveref}
\newtheorem{theorem}{Theorem}
\newtheorem{lemma}{Lemma}
\newtheorem{corollary}{Corollary}
\newtheorem{claim}{Claim}

\newenvironment{amatrix}[1]{%
  \left[\begin{array}{@{}*{#1}{c}|c@{}}
}{%
  \end{array}\right]
}

\newcommand{\bra}[1]{\left[#1\right]}
\newcommand{\pa}[1]{\left(#1\right)}
\newcommand{\set}[1]{\left\{ #1 \right\}}
\newcommand{\abs}[1]{\left|#1\right|}

\newcommand{\mc}{\mathcal}

\newcommand{\ceil}[1]{\left \lceil #1 \right \rceil}

\title{Robust Restaking Networks}
\author{Naveen Durvasula\thanks{Columbia University. Email: \texttt{nkd2126@columbia.edu}.} \and Tim Roughgarden\thanks{Columbia University \& a16z crypto. Email: \texttt{tim.roughgarden@gmail.com}.}}
\date{\today}
\allowdisplaybreaks

\begin{document}
\maketitle

\begin{abstract}
We study the risks of validator reuse across multiple services in a
restaking protocol. We characterize the robust security of a restaking
network as a function of the buffer between the costs and profits from
attacks.  For example, our results imply that if attack costs
always exceed attack profits by 10\%, then a sudden loss of .1\% of
the overall stake (e.g., due to a software error) cannot result in the
ultimate loss of more than 1.1\% of the overall stake.
We also provide local analogs of these overcollateralization conditions
and robust security guarantees that apply specifically for a target
service or coalition of services. All of our bounds on worst-case
stake loss are the best possible. Finally, we bound the
maximum-possible length of a cascade of attacks.

Our results suggest measures of robustness that could be exposed to
the participants in a restaking protocol. We also suggest
polynomial-time computable sufficient conditions that can proxy for
these measures.
\end{abstract}

\section{Introduction}

\subsection{Sharing Validators Across Services}

Major blockchain protocols such as Bitcoin or Ethereum are
``decentralized,'' meaning that transaction execution is carried
out by a large and diverse set of ``validators.'' 
Such protocols offer a form of ``trusted computation,'' in the sense
that, because they are decentralized, no one individual or entity can
easily interfere with their execution. A decentralized and
Turing-complete smart contract platform such as Ethereum can then be
viewed as a trusted programmable computer capable of performing
arbitrary computations.

While Turing-complete, the computing functionality offered by Ethereum
smart contracts suffers from limitations imposed by design decisions
in the underlying consensus protocol. Most obviously, computation and
storage in the Ethereum virtual machine is scarce, with perhaps 15--20
transactions processed per second.  Could the Ethereum protocol be
somehow bypassed, opening the door for different or more powerful
computing functionality, while retaining at least some of the
protocol's decentralization?  Or, what about applications that are not
compatible with all Ethereum validators, perhaps due to demanding
hardware requirements or regulatory constraints?

One natural approach to addressing these challenges is to allow the
{\em reuse} of a blockchain protocol's validators across multiple
services, where a ``service'' is some task that could be carried out
by some subset of validators. (The initial blockchain protocol can be viewed as the
canonical service, performed by all validators.)
For example, such services
could include alternative consensus protocols
(perhaps with higher throughput, a different virtual machine, or
different consistency-liveness trade-offs), storage (``data
availability''), or verifiable off-chain computation (``zk
coprocessors'').\footnote{%
Several projects, in various stages of production, are currently
exploring this idea; the Eigenlayer restaking protocol~\cite{el}
is perhaps the most well known of them. We stress that our goal here
is to develop a model that isolates some of the fundamental challenges
and risks of validator reuse, independent of any specific
implementation of the idea.}

The obvious danger of validator reuse is an increased risk of
a validator deviating from its intended behavior (e.g., due to overwhelming computational
responsibilities). Our focus here is deliberate validator
deviations in response to economic incentives, such as the profits that could
be obtained by corrupting one or more services. The goal of this paper is to
quantify such risks:
\begin{quote}
Under what conditions can validators be safely reused across
multiple services?
\end{quote}

\subsection{Cryptoeconomic Security}

We first review the usual ``cryptoeconomic'' approach to answering a
more basic question: when is a blockchain protocol, without any
additional services, ``safe from attack''?  The idea is to perform a
cost-benefit analysis from the perspective of an attacker, and declare
the protocol safe if the cost of carrying out an attack
exceeds the profit that the attacker can expect from it; see also Figure~\ref{f:star}.
In the
specific case of a proof-of-stake blockchain protocol with slashing
(such as Ethereum), the cost can be estimated as the value of the
validator stake that would be lost to slashing following an attack. 
For example, let~$V$ denote the set of validators of a proof-of-stake
blockchain protocol and $\sigma_v$ the stake of validator~$v \in V$
(e.g., 32 ETH in Ethereum). In a typical PBFT-type protocol in which
double-voting validators lose all of their stake, the 
cost of an attack (i.e., causing a consistency violation) can be
bounded below by $\tfrac{1}{3} \sum_{v \in V} \sigma_v$. In this case,
the protocol can be regarded as cryptoeconomically secure provided
the estimated profit~$\pi$ of an attack is less than this quantity.
To the extent that there is a ``buffer'' between~$\tfrac{1}{3} \sum_{v
  \in V} \sigma_v$ and~$\pi$, the protocol can be treated as
``robustly secure,'' meaning secure even after a sudden loss of some
amount of stake (e.g., due to slashing caused by a software error).

\begin{figure}[h]
\centering
\includegraphics[width=.425\textwidth]{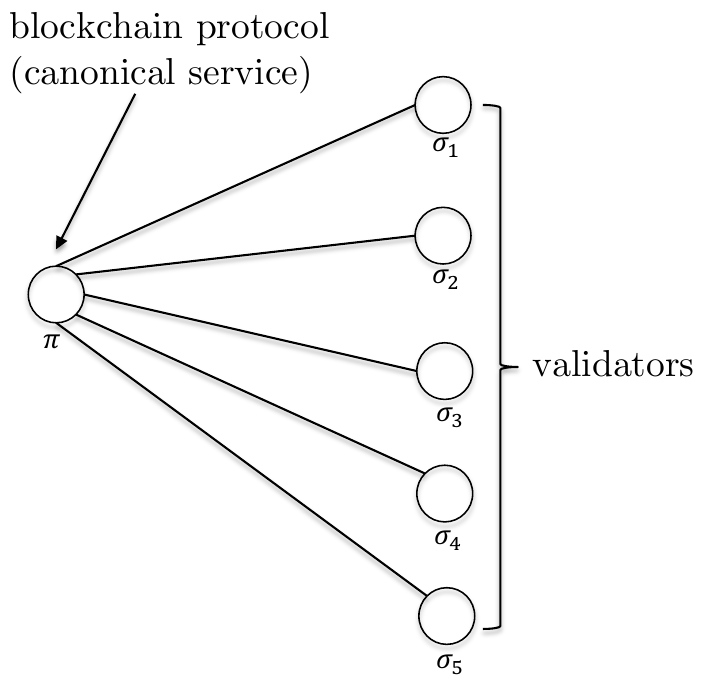}
\caption{A blockchain protocol operated by a collection of validators,
  with~$\pi$ denoting the profit of successfully attacking the
  protocol and~$\sigma_v$ the amount of stake posted by a
  validator~$v$ as collateral.}\label{f:star}
\end{figure}

\subsection{Our Results: Robustly Secure Validator Reuse}\label{ss:global}

Viewing the basic scenario of a blockchain protocol as a star graph
(with one ``service'' representing the protocol connected to all the
validators running it), the more complex scenario of validators reused
across multiple services can be viewed as an arbitrary bipartite
graph (Figure~\ref{f:bipartite}). As before, we suppose that each
validator~$v \in V$ has some
stake~$\sigma_v$ that can be confiscated in the event that~$v$
participates in an attack on a service that it has agreed to
support. There is now a set~$S$ of services, with~$\pi_s$ denoting the
profit an attacker would obtain by compromising the service~$s \in
S$.\footnote{We follow~\cite{el} and assume that the estimates on
  attack profitability (the $\pi_s$'s) are given. Developing tools to
  help produce such estimates in practice is an important open
  research direction.}
We assume that compromising a service~$s \in S$ requires the
participation of an~$\alpha_s$ fraction of the overall validator stake
supporting~$v$ (e.g., $\alpha_s = 1/3$).
With this expanded network formalism to capture multiple services,
when should we consider a network to be ``secure''?

\begin{figure}[h]
\centering
\includegraphics[width=.5\textwidth]{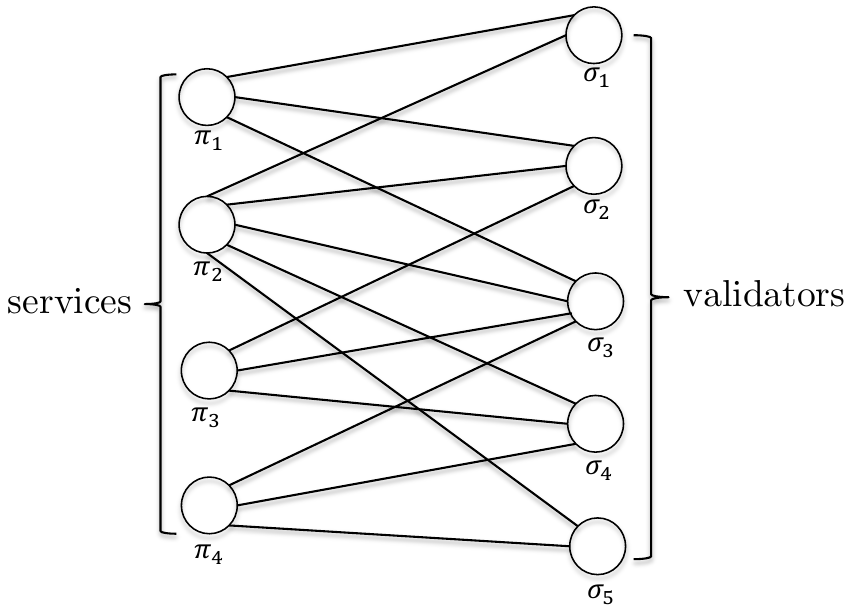}
\caption{A general restaking network, with validators reused across
  multiple services.}\label{f:bipartite}
\end{figure}

Intuitively, a network is insecure if there is a set of services that
can be profitably attacked. This requires two conditions to be
satisfied: the validators~$B \subseteq V$ carrying out the attack must
control sufficient stake to do so (for each service~$s$ in the
attacked set~$A$, the validators of~$B$ control at least a $\alpha_s$
fraction of the staked pledged to~$s$), and they must profit from it
(i.e., $\sum_{s \in A} \pi_s > \sum_{v \in B} \sigma_v$).
We call such a pair~$(A,B)$ a {\em valid attack} and a network {\em
  secure} if there are no valid attacks; see also Figure~\ref{f:secure}.
With multiple services,
security is an inherently combinatorial (as opposed to binary)
notion. For example, the computational problem of checking whether a network is secure is 
as hard as the (coNP-hard) problem of verifying the expansion of a bipartite graph (see~\cite{khot2016hardness}).

\begin{figure}
\centering
\subfigure[A valid attack]{\includegraphics[width=.4\textwidth]{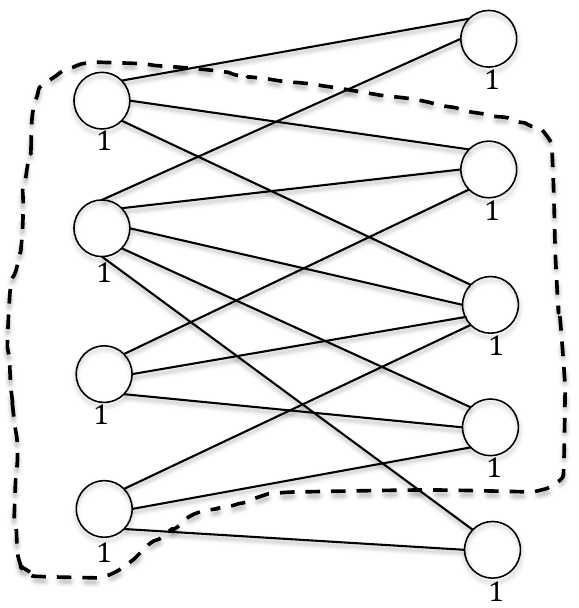}}\qquad
\subfigure[A secure network]{\includegraphics[width=.4\textwidth]{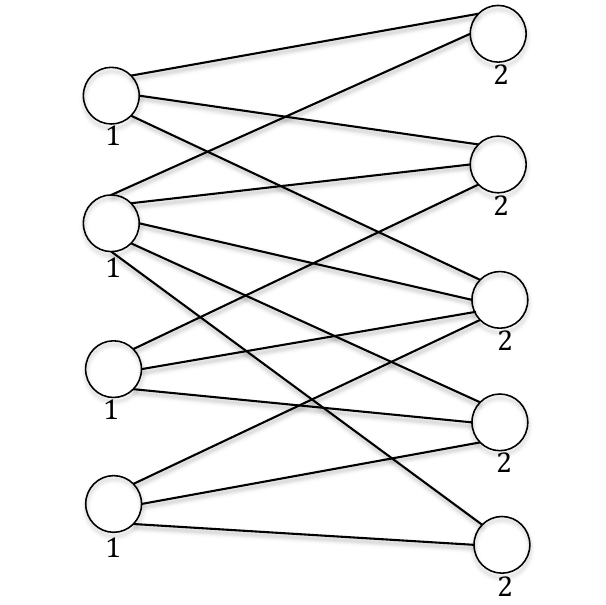}}
\caption{Two restaking networks. Each service (left-hand side vertex) and validator (right-hand
  side vertex) is labeled with its profit-from corruption or stake,
  respectively. 
Assume that a service can be corrupted if and only if it is attacked
by at least half of its validators (i.e., $\alpha_s = 1/2$ for every
  service~$s$). The restaking network in~(a) is not secure because
  there is a valid attack (indicated by the dotted line): three
  validators can earn a profit of~4 by
  corrupting all four services while losing only three units of stake.
The restaking network in~(b) is secure.}\label{f:secure}
\end{figure}

When is a network ``robustly secure,'' in the sense that the sudden
loss of a small amount of stake cannot enable a catastrophic attack?
Unlike the ``all-or-nothing'' version of this question with a single
service, with multiple services, the following more fine-grained question is the appropriate one: given an initial shock in the
form of a sudden loss of a $\psi$ fraction of the overall stake, what
is the total fraction of stake that might be lost following any
consequent valid attacks?

As in the case of a single service, some amount of ``buffer'' in stake
(relative to attack profits) is necessary for robust security. We
parameterize this overcollateralization factor via a parameter~$\gamma$ and suppose that,
whenever~$B \subseteq V$ is a subset of validators capable of
corrupting all the services in $A \subseteq S$, the total stake
of~$B$ is at least a $1+\gamma$ factor larger than the total profit from
corrupting all of~$A$. For example, this condition holds in the
network in Figure~\ref{f:secure}(b) for~$\gamma=1/2$ (but not for larger
values of~$\gamma$).

Our first main result (\cref{thm:slack}) precisely characterizes the
worst-case (over bipartite graphs and shocks, as a function
of~$\gamma$) fraction of the overall
stake that can be lost due to a shock
of size $\psi$: $\pa{1+\tfrac{1}{\gamma}}\psi$. Because the network
was secure prior to the shock, the value of a
$\pa{1+\tfrac{1}{\gamma}}\psi$ fraction of the overall stake is also
an upper bound on the total profit obtained from all of the attacked
services.

We also show that our result is tight
in a strong sense (\cref{thm:gn}, \cref{thm:noslack}, and
\cref{thm:localtight2}): for every $\psi$, $\gamma$, and $\epsilon$
greater than zero such that $0 \le \pa{1 + \frac1\gamma}\psi - \epsilon \le 1$,
there exists a restaking graph in which a $\psi$
fraction of the overall stake can disappear in a shock that results in
the loss of at least a $\pa{1 + \frac1\gamma}\psi - \epsilon$ fraction
of the overall stake.\footnote{After slightly reducing the validator
  stakes, the network
  in Figure~\ref{f:secure}(b) already shows that the bound is tight for
  the special case in which $\psi = 1/5$ and $\gamma$ is arbitrarily
  close to $1/2$. (Consider a shock that knocks out the validator that
is connected to all four services.)}

Qualitatively, this result implies that a constant-factor
strengthening of the obvious necessary condition for security
automatically implies robust security. For example, if attack costs
always exceed attack profits by 10\%, then a sudden loss of .1\% of
the overall stake cannot result in the ultimate loss of more than
1.1\% of the overall stake. 

Our result suggests a ``risk measure'' that could be exposed to the
participants in a restaking protocol, namely the maximum value of the
buffer parameter~$\gamma$ that holds with respect to the current
restaking network. We also suggest easy-to-check
sufficient conditions that can proxy for this risk measure (\cref{cor:ELslack}). These conditions are similarly tight, as shown in \cref{thm:gn}, \cref{thm:noslack}, and  \cref{thm:localtight2}.

\subsection{Our Results: Local Robust Security Guarantees}\label{ss:local}

\newcommand{\sse}{\subseteq}

The results described in Section~\ref{ss:global} are ``global'' with
respect to the network structure, in three distinct senses: (i) the
overcollateralization condition is assumed to hold for every
subset~$B \sse V$ of validators and~$A \sse S$ of services that~$B$ is
capable of corrupting; (ii) the initial shock can affect any subset of
validators, subject to the assumed bound of $\psi$ on the total
fraction of stake lost; and (iii) any subset of validators might lose
stake following the initial shock, subject to our upper bound of
$(1+\tfrac{1}{\gamma})\psi$ on the total fraction of lost stake.

\vspace{-.1in}

\paragraph{Local guarantees.}
We next pursue more general ``local'' guarantees, which are
parameterized by a set~$C$ of services. (The global guarantees will
correspond to the special case in which~$C=S$.)
For example, $C$ might be a set of closely related services that share
a number of dedicated validators. The operators of such a set~$C$
might object to both the assumptions and the conclusion of our global
guarantee:
\begin{itemize}

\item How can we be sure that the overcollateralization factor holds for
  services and validators that we know nothing about?

\item And even if we could, 
how can we be sure that random validators that we have nothing
  to do with won't suddenly lose their stake (e.g., because they
  supported a malicious or buggy service), resulting in an initial
  shock the causes the loss of more than a $\psi$ fraction of the
  overall stake?

\item And even if we could, how can we be sure that our validators
  won't be the ones that lose their stake following a shock that is
  purely the fault of other services and/or validators?

\end{itemize}

To address these concerns, we next consider a refined version of the
basic model, parameterized by a set~$C$ of services.  We denote
by~$\Gamma(C)$ the validators that are exclusive to~$C$, meaning that
they contribute to no services outside~$C$. (In the special case in
which~$C=S$, $\Gamma(C)=V$ and we recover the original model.)
Intuitively, the validators in~$\Gamma(C)$ are the ones that services
in~$C$ are ``counting on''; other validators support (potentially
malicious or buggy) services outside~$C$ and, from $C$'s perspective,
could disappear at any time.  We then restrict attention to initial
shocks in which at most a~$\psi$ fraction of the total stake
controlled by the validators of~$\Gamma(C)$ is lost. (The shock can affect
validators outside of~$\Gamma(C)$ arbitrarily.)  The goal is then to
identify overcollateralization conditions guaranteeing that, no matter
what the initial shock and subsequent attacks, the fraction of stake
ultimately lost by the validators in~$\Gamma(C)$ is bounded (by a
function of the shock size~$\psi$ and an overcollateralization
parameter~$\gamma$).

Generalizing our global guarantees to local guarantees is not
straightforward and, as the next two examples show, requires
additional compromises. The first example shows that protection can be
guaranteed only against a subset of valid attacks, a natural and well
motivated subset that we call ``stable attacks.'' The second example
shows that, even when restricting attention to stable attacks,
overcollateralization is required not only for potential attacks, but
more generally for what we call ``attack headers.''

\begin{figure}[ht]
    \centering
    \begin{minipage}{.5\textwidth}
    \centering
    \includegraphics[scale=0.25]{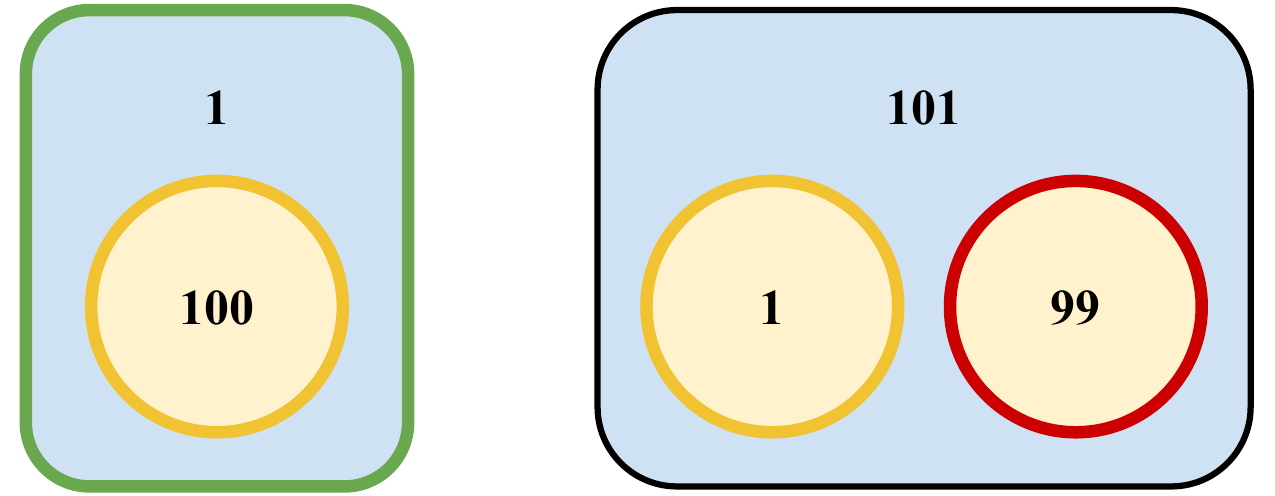}
    \end{minipage}\begin{minipage}{.5\textwidth}
    \centering
    \includegraphics[scale=0.25]{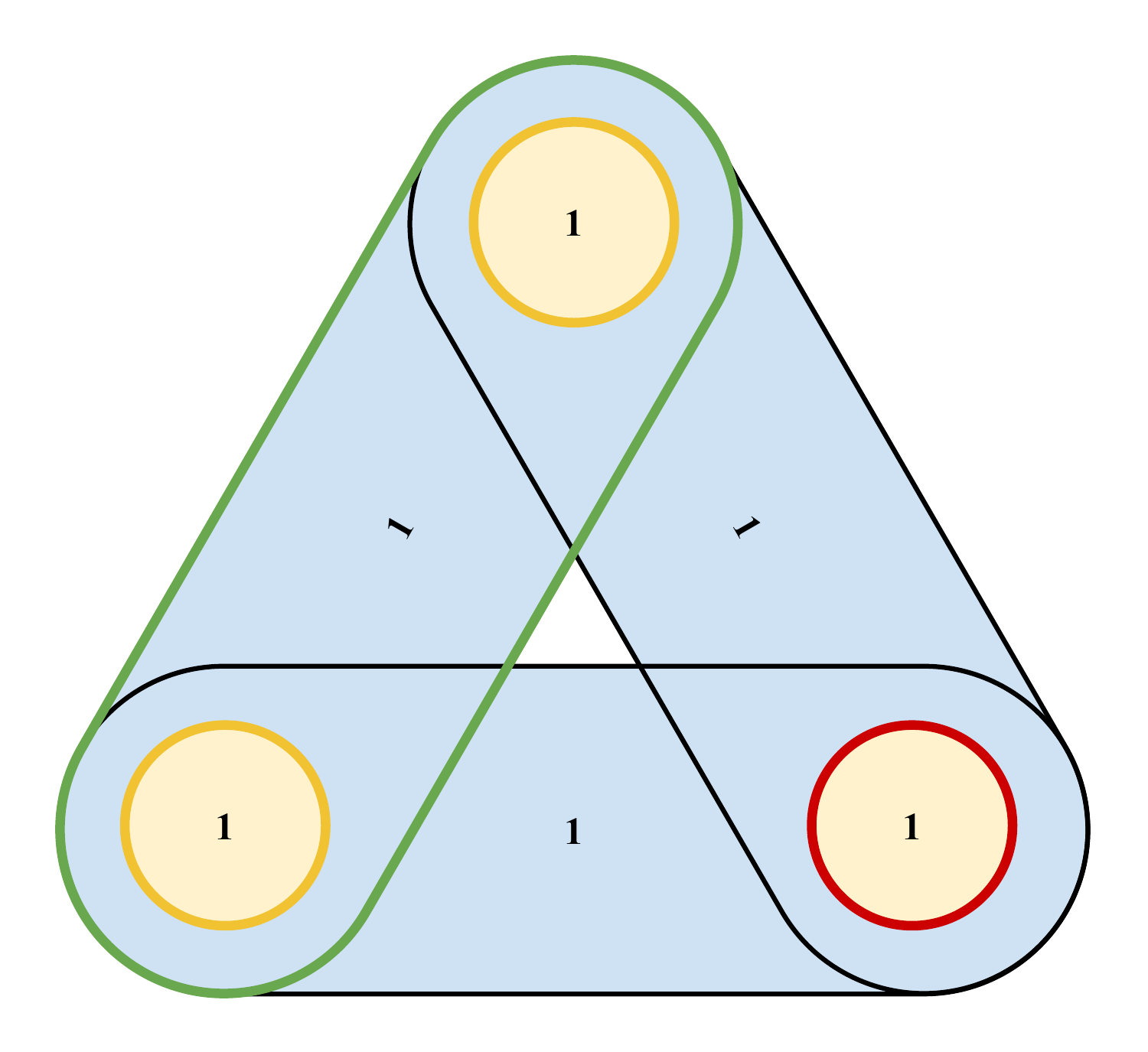}
    \end{minipage}
    \caption{\textbf{Simple overcollateralization is insufficient in
        the local setting.} There are two restaking networks shown
      above. In each, the validators are denoted along with their
      corresponding stakes by the yellow circles. Services and their
      profits from corruption are denoted by the blue rounded
      squares. In each of these networks, the service outlined in
      green is overcollateralized. However, despite being unrelated to
      the overcollateralized service, if the validator outlined in red
      disappears, the validators in yellow can attack all services
      (including the overcollateralized one).}
    \label{fig:localdemo}
\end{figure}

\vspace{-.1in}

\paragraph{Necessity of restricting to stable attacks.}
In more detail, in the first example (depicted on the left of
\cref{fig:localdemo}), we
consider a restaking network with two services and three
validators. (Vertices in this figure correspond to validators, which
were previously represented as the right-hand side vertices of a
bipartite graph; each service is now identified with its neighborhood
of validators.) 
The service highlighted in green on the left has $100$
times as much stake securing it as its profit from corruption, and
shares no validators with other services, so one might hope that it
would be protected from any shocks that affect only the other
validators.
However, if the validator
outlined in red disappears, the two validators outlined in yellow
can attack both services for
a profit of $102-101 > 0$.
But this example is unsatisfying: the validator with stake $1$ could
have attacked the service with profit from corruption $101$ on its own
to yield a profit of $100$. The addition of the validator with stake
$100$ added $100$ to the cost of the attack, but only added $1$ (from
the corruption of the service highlighted in green) to the profit. We
show in \cref{thm:localimpossible} that the issue suggested by this
example is fundamental: without further restrictions on attacks that
rule out contrived examples such as this one, there does not exist a
local condition that guarantees local security.
In response, we confine
attention to sequences of what we call {\em stable} attacks in which
all of the attacking validators contribute positively to the profit of
the attack (as opposed to free riding on the profits attributable to
other attacking validators). The attack in the example is not
stable, as the validator with stake $100$ was not a profitable
addition to the attack that could have been carried out by the
validator with stake $1$.

\vspace{-.1in}

\paragraph{Necessity of overcollateralizing attack headers.}
Even if we restrict our attention to stable attacks, simple
overcollateralization is insufficient to guarantee local security. To
appreciate the issue, consider the restaking network shown on the
right in \cref{fig:localdemo}. In this example, there are three
services, each with a profit from corruption of $1$, and three
services each with stake $1$. Each validator is used to secure two
different services. The service outlined in green is
overcollateralized in that its profit from corruption is $1$, but two
units of stake stake are securing it. Despite this, if the unrelated
validator outlined in red disappears, the two validators outlined in
yellow can attack all three services for a profit of $3-2 > 0$.
Furthermore, if all validators are required to attack each service
(i.e., $\alpha_s=1$), then this attack is stable, as the inclusion
of both validators in the attack is profitable. Thus, despite being
overcollateralized, a stable attack can be launched on the service
highlighted in green even if stake unrelated to the service disappears
in a shock.

We find that to guarantee robust security for a coalition of services
$C$ using only ``local'' information, it is necessary to
overcollateralize not only pairs~$(A,B)$ where~$B$ is a set of
validators capable of attacking all services in~$A$, but a more
general collection of pairs that we call {\em attack
  headers}. Informally, this amounts to requiring that there is some
``buffer'' in stake for any potential attack on some services in $C$
even if we were to allow every validator that is also securing a
service outside of $C$ to join the attack without considering their
profitability. Our main result here (\cref{thm:local}) formally
provides a local condition guaranteeing that, whenever an initial
shock knocks out at most a $\psi$ fraction of the stake that provides
security exclusively to~$C$, the worst-case loss of such stake, after
an arbitrary sequence of stable attacks, is at most a
$(1+\tfrac{1}{\gamma})\psi$ fraction. We show that our bounds are
tight and indeed require overcollateralization of all attack headers
(\cref{cor:localtight}, \cref{thm:globalinsufficient}, and
\cref{thm:localtight2}), and again provide easily computable
sufficient (and similarly tight) conditions that can proxy for the
overcollateralization condition (\cref{cor:local}).  Our local condition
generalizes the global overcollateralization condition, with the latter
corresponding to the special case of the former in which~$C=S$.

\subsection{Our Results: Cascading Attacks}

An attack on a restaking network results in a loss of stake (of the
attacking validators), and this may introduce new opportunities for
other sets of validators to carry out profitable attacks. That is, an
initial shock may set off an entire {\em cascade} of attacks. (All of
the bounds on stake loss described in Sections~\ref{ss:global}
and~\ref{ss:local} hold for cascades of attacks of arbitrary length.)
Our final result concerns the maximum-possible length of such a
cascade, and shows that this quantity is also governed in part by the
overcollateralization factor~$\gamma$.

Precisely, we define the {\em reference depth} of a cascade of attacks
as a measure of the ``long-range dependence'' between different
attacks in an attack sequence. For example, if each attack in the
sequence is directly enabled by the loss of the validators slashed in
the previous attack, then the reference depth of the sequence
is~1. Our main result here (\cref{thm:length}) bounds the
maximum-possible attack length as a function of the reference
depth~$k$, the shock size~$\psi$, the overcollateralization
factor~$\gamma$, and the minimum stake $\epsilon$ held by a validator:
$k ( 1 + \log_{1 + \gamma}(\frac{\psi \cdot [\text{total
    stake}]}{\epsilon \gamma}))$.
For example, in the case of constant reference depth and equal
validator stake amounts, the worst-case attack length is logarithmic
in the number of validators with overcollateralization and linear
without it.

\subsection{Related Work}

Our focus on the risks of cascading failures following a small shock
echoes some of themes in the literature on systemic risk in financial
networks. For example, Eisenberg and Noe~\cite{eisenberg2001systemic}
that study the existence and structure of inter-firm payments in a
financial network following a default.  This work is extended by
Glasserman and Young~\cite{glasserman2016contagion}, who study how
``connectedness,'' meaning the fraction of liabilities that a firm
externally owes, affects contagion risk.  Acemoglu et
al.~\cite{acemoglu2015systemic} build further on this work and study
network ``stability,'' meaning the propensity for shocks to propagate;
they show that connectivity initially improves stability, but then at
a phase transition, denser connectivity leads to increased shock
propagation.  In a subsequent paper, Acemoglu et
al.~\cite{acemoglu2015networks} unify a number of the preceding
results. 

A separate line of work, beginning with Chen et
al.~\cite{chen2013axiomatic}, aims to axiomatically characterize
systemic risk measures.  The model in~\cite{chen2013axiomatic} can
capture, in particular, a contagion model characterized by a matrix of
profits and losses over different firms and outcomes. The results
in~\cite{chen2013axiomatic} characterize the global measures of risk
(operating on the matrix) that satisfies certain sets of desirable
axioms.  This is work is expanded upon in Kromer et
al.~\cite{kromer2016systemic}, where the authors consider general
outcome measure spaces, as well as in Feinstein et
al.~\cite{feinstein2017measures}, where the authors consider
set-valued risk measures.  Battison et al.\cite{battiston2016price}
study systemic risk measurement when a regulator has limited
information about contracts made between financial network
participants (e.g., the fraction of the face value that can likely be
recovered if a counterparty defaults), and show how small errors in
knowledge can lead to large errors in systemic risk measurement.

Our work also shares some conceptual similarity with the well-known
work of Diamond and Dybvig~\cite{diamond1983bank} on bank runs and
of Brunnermeier et al.~\cite{brunnermeier2012risk} on the risks of re-hypothecation.

The model in the present work differs substantially from those
considered in the aforementioned papers, in large part because of the
idiosyncrasies of restaking networks, including their combinatorial
and bipartite nature and their susceptibility to economically
motivated attacks.

Restaking has also, to a limited extent, been studied in its own
right. The EigenLayer team introduces restaking in \cite{el}. The
framing of an economically motivated attacker that trades off stake
loss with profits from corruption appears in \cite{el}, and is used
also by Deb et al.~\cite{deb2024stakesure}.  Chitra and
Neuder~\cite{blogpost} discuss restaking risk from a validator
perspective, comparing restaking with investments in other financial
instruments (e.g. bonds).  Alexander~\cite{alexander2024leveraged}
considers the interplay between existing leveraging schemes and liquid
restaking tokens, which can amplify large-scale credit risk.

\section{Model}\label{sec:model}

\paragraph{Validators and Services.} We consider a setting in which there is a set $V$ of validators and a set $S$ of services. Each service $s \in S$ has some profit from corruption $\pi_s$, and each validator $v \in V$ has some stake $\sigma_v$. We also associate with each service $s$ a parameter $\alpha_s$ that denotes the fraction of stake required to corrupt/launch an attack on $s$. We call a bipartite graph $G = (S, V, E, \pi, \sigma, \alpha)$ a \textit{restaking graph}; an edge is drawn between a validator $v \in V$ and a service $s \in S$ if $v$ is restaking for $s$. For a given set of vertices $A$ in a graph $G$, we use the notation $N_G(A)$ to denote the neighbors of $A$. 

\paragraph{Attack Dynamics.} For simplicity, we assume that validators lose their full stake $\sigma_v$ if they launch an attack on a service. As such, for a given collection of services $A \subseteq S$, and a given collection of validators $B \subseteq V$ restaking for those services, we say that $(A, B)$ is an \textit{attacking coalition} for a restaking graph $G$ if the validators in $B$ possess enough stake to corrupt the services $A$:
\begin{align}
    \underbrace{\sum_{v \in B \cap N_G(\set{s})} \sigma_v}_{\text{Total stake in $s$ owned by validators $B$}} \ge \alpha_s \cdot \underbrace{\sum_{v \in N_G(\set{s})}\sigma_v}_{\text{Total amount restaked in $s$}} && \forall s \in A\label{eq:threat}
\end{align}
We further say that $(A, B)$ is a \textit{valid attack} if it is an attacking coalition that has an incentive to launch an attack:
\begin{equation}
    \underbrace{\sum_{s \in A} \pi_s}_{\text{Total profit from corrupting $A$}} > \underbrace{\sum_{v \in B} \sigma_v}_{\text{Total stake owned by validators $B$}}\label{eq:balance}
\end{equation}
If a valid attack $(A, B)$ is carried out, we denote by $G \searrow B$ the induced subgraph $G\bra{S, V \setminus B}$. The graph $G \searrow B$ denotes the state of the restaking graph after the attack is carried out. If no valid attacks exist on the graph, we call it \textit{secure}. To simplify notation, for any subset of validators $B \subseteq V$, we will use the shorthand $\sigma_B$ to denote $\sum_{v \in B} \sigma_v$. Similarly, for any $A \subseteq S$, we will use $\pi_A$ as shorthand for $\sum_{s \in A}\pi_s$.

\paragraph{EigenLayer sufficient conditions.} We note in passing that, in their whitepaper \cite{el}, EigenLayer proposes some efficiently verifiable sufficient conditions for security that they check to ensure that an attack does not exist.
\begin{claim}[EigenLayer Sufficient Conditions, from Appendix B.1 of the EigenLayer Whitepaper \cite{el}]\label{claim:EL}
A restaking graph $G$ is secure if for each validator $v \in V$,
\begin{equation}\label{eq:EL}
    \sum_{s \in N_G\set{v}} \frac{\sigma_v }{\sigma_{N_G\pa{\set{s}}}} \cdot \frac{\pi_s}{\alpha_s} \le \sigma_v 
\end{equation}
\end{claim}
\begin{proof}
It is shown in Appendix B.1 of the whitepaper that if \cref{eq:EL} holds for $G$, then the graph is secure (i.e. no valid attacks $(A,B)$ satisfying \cref{eq:threat,eq:balance} exist).
\end{proof}

\paragraph{Cascading attacks.} Our goal is to understand when a small shock can result in the loss of a large fraction of the overall stake. Small shocks can turn into large shocks by means of a cascading attack. Formally, we say that a disjoint sequence $(A_1, B_1), \dots, (A_T, B_T) \in 2^S \times 2^V$ is a {\em valid cascade of attacks} on a restaking graph $G = (S, V, E, \pi, \sigma, \alpha)$ if for each $t \in [T]$, $(A_t, B_t)$ is a valid attack on $G \searrow \bigcup_{i = 1}^{t-1}  B_i$. We denote by $\mc C(G)$ the set of all such sequences of valid cascading attacks. 

\paragraph{Worst-case stake loss.} We now define a metric that measures the total potential loss of stake due to a sequence of cascading attacks. In our model, we first suppose that an initial small shock decreases the amount of stake. Formally, we define, for a given restaking graph $G = (S, V, E, \pi, \sigma, \alpha)$,

\begin{equation}
    \mathbb D_\psi(G) := \set{D \subseteq V \mid \frac{\sigma_D}{\sigma_V} \le \psi}
\end{equation}
to be the set of validator coalitions that constitute at most a $\psi$-fraction of all stake. Given some $D \in \mathbb D_\psi(G)$, we use the notation $G \searrow D := G\bra{S, V \setminus D}$ to denote the induced subgraph of the restaking graph when we delete the validators $D$.
%(i.e. the state of the graph after a small shock in the amount of restaked Ethereum).  
We now define %the cascade multiplier as  
\begin{equation}\label{eq:multiplier}
R_\psi(G) := \underbrace{\psi}_{\text{Initial shock}} + \max_{D \in \mathbb D_\psi(G)} \max_{(A_1, B_1), \dots, (A_T, B_T) \in \mc C(G\searrow D)} \underbrace{\frac{\sigma_{\bigcup_{t = 1}^T B_t}}{\sigma_V}}_{\text{Stake lost from cascading attacks}}
\end{equation}

This quantity represents the worst-case total fraction of stake lost due to an initial $\psi$-fraction of the stake disappearing. By construction, $\psi \le R_\psi(G) \le 1$.

\section{Overcollateralization Provides Robust Security}\label{sec:global_ub}

In this section, we show that ``scaling up'' the definition of security automatically results in robust security, meaning bounded losses from cascading attacks that follow an initial shock. We first show that, without loss of generality, it suffices to consider single valid attacks $(A, B) \in \mc C(G\searrow D)$ instead of more general cascading attacks.

\begin{lemma}\label{lem:acombine}
    Let $G = (S, V, E, \pi, \sigma, \alpha)$ be an arbitrary restaking graph, and further suppose that $(A, B)$ is an attacking coalition on $G \searrow D$, where $D \subseteq V$. Then, $(A, B \cup D)$ is an attacking coalition on $G$. 
\end{lemma}
\begin{proof}
    Because $(A, B)$ is an attacking coalition on $G \searrow D$, we must have by \cref{eq:threat} that
    \begin{align}
        \sigma_{B \cap N_G\set{s}} \ge \alpha_s \cdot \sigma_{N_G \set{s} \setminus D} && \forall s \in A
    \end{align}
    It follows that for any $s \in A$,
    \begin{align}
    \sigma_{\pa{B \cup D} \cap N_G\set{s}} &= \sigma_{B \cap N_G\set{s}} + \sigma_{D \cap N_G\set{s}}\\
    &\ge \alpha_s \cdot \sigma_{ N_G\set{s} \setminus D}  + \sigma_{D \cap N_G\set{s}}\\
    &\ge \alpha_s \cdot \sigma_{N_G\set{s}}
\end{align}
and the desired result follows.
\end{proof}
\begin{corollary}\label{cor:combine}
Let $G = (S, V, E, \pi, \sigma, \alpha)$ be an arbitrary restaking graph, and further suppose that $(A_1, B_1), \dots, (A_T, B_T) \in \mc C(G)$ is a valid sequence of cascading attacks on $G$. Then, $\pa{\bigcup_{t = 1}^T A_t, \bigcup_{t = 1}^T B_t}$ is also a valid attack on $G$. 
\end{corollary}
\begin{proof}
    By repeatedly applying Lemma \ref{lem:acombine}, we find that for each $t \in [T]$, $\pa{A_t, \bigcup_{i = 1}^t B_t}$ is an attacking coalition on $G$. It follows by inspection of \cref{eq:threat} that we must therefore have that $\pa{\bigcup_t A_t, \bigcup_t B_t}$ is an attacking coalition on $G$. To finish the result, it suffices to show that \cref{eq:balance} holds for this attacking coalition on the original graph $G$. This follows from the disjointness of the $A_t$'s and of the $B_t$'s:
    \begin{equation}
    \pi_{\bigcup_t A_t} = \sum_{t = 1} \pi_{A_t} > \sum_{t = 1}^T \sigma_{B_t}  =  \sigma_{\bigcup_t B_t}
    \end{equation}
    where in the inner inequality we use that for each $t \in [T]$, $\pi_{A_t} > \sigma_{B_t} $ by \cref{eq:balance}, as $(A_t, B_t)$ is a valid attack on $G \searrow \bigcup_{i=1}^{t-1} B_i$. It follows that $\pa{\bigcup_t A_t, \bigcup_t B_t}$ is a valid attack on $G$.
\end{proof}

\paragraph{Adding Multiplicative Slack.} Our condition is given by adding multiplicative slack to \cref{eq:balance}. Formally, we say that a restaking graph $G$ is {\em secure with $\gamma$-slack} if for all attacking coalitions $(A, B)$ on $G$,
\begin{equation}\label{eq:slack}
    (1 + \gamma) \underbrace{\pi_{A}}_{\text{Total profit from corrupting $A$}} \le \underbrace{\sigma_{B}}_{\text{Total stake owned by validators $B$}}
\end{equation}

\begin{theorem}\label{thm:slack}
    Suppose that a restaking graph $G = (S, V, E, \pi, \sigma, \alpha)$ is secure with $\gamma$-slack for some $\gamma > 0$. Then, for any $\psi > 0$, $R_\psi(G) < \pa{1 + \frac1\gamma} \psi$.
\end{theorem}

\begin{proof}
Take any $\psi > 0$ and any $D \in \mathbb D_\psi(G)$ for some restaking graph $G$ where (\ref{eq:slack}) holds. Let $(A_1, B_1) \dots, (A_T, B_T) \in \mc C(G\searrow D)$ be arbitrary. Applying \cref{cor:combine}, we must have that $(\bigcup_t A_t, \bigcup_t B_t) \in \mc C(G\searrow D)$ as well. Defining $A := \bigcup_t A_t$ and $B := \bigcup_t B_t$, we must therefore have that $(A, B)$ is an attacking coalition on $G \searrow D$, and furthermore that
\begin{align}
    \pi_A &> \sigma_B \label{eq:cascattack}
\end{align}
 By \cref{lem:acombine}, we must also have that $(A, B \cup D)$ is an attacking coalition on the original graph $G$. It then follows that as $G$ is secure with $\gamma$-slack, \cref{eq:slack} must hold on $(A, B \cup D)$, whence
\begin{equation}
    (1 + \gamma)  \pi_A \le  \sigma_{B \cup D} = \sigma_{B} + \sigma_{D}
\end{equation}
Putting this together with \cref{eq:cascattack}, we find that
\begin{equation}
    (1 + \gamma) \sigma_B < (1 + \gamma) \pi_A \le \sigma_B + \sigma_D  \le  \sigma_B + \psi \cdot \sigma_V
\end{equation}
It follows that 
\begin{align}
      \gamma \cdot \sigma_B < \psi \cdot \sigma_V &\implies \frac{\sigma_B }{\sigma_V} <  \frac{\psi}{\gamma}\\
      &\implies \psi +  \frac{\sigma_B }{\sigma_V} < \pa{1 + \frac1\gamma}\psi 
\end{align}
As we took $A_1, \dots, A_T$ to be arbitrary, we find that $R_\psi(G) < \pa{1 + \frac1\gamma}\psi$, as desired.
\end{proof}

The EigenLayer sufficient conditions (\ref{eq:EL}) can be similarly ``scaled up'' to yield efficiently checkable sufficient conditions for security with $\gamma$-slack (and hence, by \cref{thm:slack}, robustness to cascading attacks).
\begin{corollary}\label{cor:ELslack}
Let $G$ be a restaking graph such that, for all validators $v \in V$,
    \begin{equation}\label{eq:ELslack}
    \sum_{s \in N_G(\set{v})} \frac{\sigma_v }{\sigma_{N_G\pa{\set{s}}} } \cdot \frac{(1 + \gamma) \pi_s}{\alpha_s} \le \sigma_v 
\end{equation}
Then, $R_\psi(G) < \pa{1 + \frac1\gamma} \psi$.
\end{corollary}
\begin{proof}
    This follows from \cref{thm:slack} and \cref{claim:EL}. Noting that the $\gamma$-slack condition holds precisely iff no valid attacks exist when profits from corruption $\pi_s$ are inflated by a multiplicative factor of $(1+\gamma)$, it suffices to apply EigenLayer's sufficient conditions from \cref{claim:EL} with modified profits from corruption $(1+\gamma)\pi_s$.
\end{proof}
Note that, given a restaking network, it is straightforward to compute the minimum value of~$\gamma$ such that the condition in~\eqref{eq:ELslack} holds. % (if nothing else, via binary search over~$\gamma$). 
This value can then be interpreted as an easily computed ``risk measure'' of such a network.

\section{Lower Bounds for Global Security}\label{sec:global_lb}

In this section, we show that the upper bounds from the previous section are tight. We first show that if there is no multiplicative slack (i.e. $\gamma = 0$), then very small shocks can cause all stake to be lost in the worst case. This holds even under the EigenLayer conditions (\ref{eq:EL})\footnote{The graph exhibited in the proof of \cref{thm:gn} has $\pi_x/\sigma_a \to \infty$ as $\epsilon \to 0$. It is possible to construct a counterexample with similar properties while maintaining that $\pi_s/\sigma_v$ is greater than some universal constant for any $s \in S$ and $v \in V$. This is done in \cref{thm:localtight2}.}.
\begin{theorem}\label{thm:gn}
    For any $0 < \epsilon < 1$, there exists a restaking graph $G$ that is secure and meets the EigenLayer condition (\ref{eq:EL}), but has $R_{\psi}(G) = 1$ for all $\psi \ge \epsilon$. 
\end{theorem}
\begin{proof}
    We construct a restaking graph $G = (S, V, E, \pi, \sigma, \alpha)$ with one service $S = \set{x}$ and two validators $V = \set{a,b}$, where an edge exists between each validator and the service (i.e. $E := (\set{x}, \set{a}), (\set{x}, \set{b})$). We then let $\sigma_a := \epsilon$, $\sigma_b := 1 - \epsilon$, $\pi_x := 1$, and $\alpha_x := 1$. Without loss of generality, we may assume $\psi < 1$ as $R_1(G) = 1$ for any restaking graph $G$. This graph satisfies (\ref{eq:EL}) as 
    \begin{align}
        \frac{\sigma_a}{\sigma_a + \sigma_b} \cdot \frac{\pi_x}{\alpha_x} = \sigma_a\\
        \frac{\sigma_b}{\sigma_a + \sigma_b} \cdot \frac{\pi_x}{\alpha_x} = \sigma_b
    \end{align}
    whence the graph is also secure by \cref{claim:EL}. We now consider an initial shock $D = \set{a}$. As $\sigma_a/\sigma_V = \epsilon \le \psi$, it follows that $D \in D_\psi(G)$. The pair $(\set{x}, \set{b})$ is a valid attack on $G \searrow D$, since $\set{b} = N_{G \searrow D} \set{x}$ whence it is an attacking coalition, and $\sigma_b < \pi_x$. It follows that $R_\psi(G) \ge \frac{\sigma_a + \sigma_b}{\sigma_V} = 1$ as desired.
\end{proof}

Next, we show that the bound we give in \cref{thm:slack} (indeed, more strongly, the condition given in \cref{cor:ELslack}) is tight for all $\psi, \gamma > 0$. 

\begin{theorem}\label{thm:noslack}
    For any $\psi, \gamma, \epsilon > 0$ such that
    \begin{equation}\label{eq:nontrivial}
        0 \le \pa{1 + \frac1\gamma}\psi - \epsilon \le 1,
    \end{equation}
    there exists a restaking graph $G$ that satisfies the condition (\ref{eq:ELslack}) from \cref{cor:ELslack} but has $R_\psi(G) \ge \pa{1 + \frac1\gamma}\psi - \epsilon$.
\end{theorem}
\begin{proof}
    We construct a restaking graph $G = (S, V, E, \pi, \sigma, \alpha)$ with three validators $V = \set{a, b, c}$ and one service $S = \set{x}$, where an edge exists between each of the validators $a$ and $b$ and the service $x$ (i.e. the edge set $E := \set{(x,a), (x,b)}$). Without loss of generality, suppose that $\epsilon \le \psi/\gamma$. Let $\sigma_a > 0$ be any positive constant. We define
    \begin{align}
        \sigma_b &:= \sigma_a \pa{\frac1\gamma - \frac\epsilon\psi}\\
        \sigma_c &:= \sigma_a\pa{\frac{1 - \psi + \epsilon}{\psi} - \frac1\gamma}\\
        \pi_x &:= \frac{\pa{1 + \frac1\gamma} \psi - \epsilon}{1 + \gamma} \cdot \sigma_V\\
        \alpha_s &:= 1
    \end{align}
    Notice first that $\sigma_b \ge 0$ as we have taken $\epsilon \le \psi/\gamma$. Next, observe that
    \begin{align}
        \sigma_c \ge 0 &\iff \frac{1 - \psi + \epsilon}{\psi} \ge \frac1\gamma\\
        &\iff 1 \ge \pa{1 + \frac1\gamma}\psi - \epsilon
    \end{align}
    whence $\sigma_c \ge 0$ by \cref{eq:nontrivial}. Finally, we must also have that $\pi_x \ge 0$ by \cref{eq:nontrivial} as well. Next, notice by construction that
    \begin{align}
        \sigma_V &= \sigma_a + \sigma_b + \sigma_c = \sigma_a\pa{1 + \frac1\gamma - \frac\epsilon\psi + \frac{1 - \psi + \epsilon}{\psi} - \frac1\gamma} = \frac{\sigma_a}{\psi}
    \end{align}
    This graph meets condition (\ref{eq:ELslack}) from \cref{cor:ELslack}, as 
    \begin{align}
        \sum_{s \in N_G\set{a}} \frac{\sigma_a}{\sigma_{N_G\set{s}}} (1 + \gamma)\pi_s &= \frac{\sigma_a}{\sigma_a + \sigma_b} \bra{\pa{1 + \frac1\gamma}\psi - \epsilon}\sigma_V\\
        &= \frac{\sigma_a}{\sigma_a\pa{1 + \frac1\gamma - \frac\epsilon\psi}}\bra{\pa{1 + \frac1\gamma}\psi - \epsilon}\frac{\sigma_a}{\psi}\\
        &= \sigma_a
    \end{align}
    A similar argument shows that 
    \begin{equation}
         \sum_{s \in N_G\set{b}} \frac{\sigma_b}{\sigma_{N_G\set{s}}} (1 + \gamma)\pi_s = \frac{\sigma_b}{\sigma_a\bra{\pa{1 + \frac1\gamma}\psi - \epsilon}}\bra{\pa{1 + \frac1\gamma}\psi - \epsilon}\sigma_a = \sigma_b
    \end{equation}
    Finally, as the validator $c$ has no neighbors, it also satisfies the condition, whence the graph indeed satisfies (\ref{eq:ELslack}). We now consider an initial shock $D := \set{a}$ to the graph. Because
    \begin{equation}
        \frac{\sigma_a}{\sigma_V} = \frac{\sigma_a}{\sigma_a/\psi} = \psi,
    \end{equation}
    the shock $D \in \mathbb D_\psi$ constitutes a $\psi$ fraction of the total stake as desired. The attack $(\set{x}, \set{b})$ is a valid attack on $G \searrow D$. To see this, notice first that $(\set{x}, \set{b})$ is an attacking coalition on $G \searrow D$ as $\set{b} = N_{G \searrow D} \set{x}$. Furthermore, $\set{b}$ is incentivized to attack since
    \begin{align}
        \pi_x - \sigma_b &= \frac{\pa{1 + \frac1\gamma} \psi - \epsilon}{1 + \gamma} \cdot \sigma_V - \sigma_b\\
        &= \bra{\frac{1 + \frac1\gamma - \frac{\epsilon}{\psi}}{1 + \gamma} - \pa{\frac1\gamma - \frac\epsilon\psi}}\sigma_a\\
        %&= \bra{\frac{\pa{1 + \frac1\gamma}\psi - \epsilon}{(1 + \gamma) \psi} - \frac1\gamma + \frac{(1+\gamma)\epsilon}{(1+\gamma)\psi}}\sigma_a\\
        &= \bra{\frac{\pa{1 + \frac1\gamma}\psi + \gamma \epsilon}{(1 + \gamma) \psi} - \frac1\gamma}\sigma_a\\
       % &= \bra{\frac{\gamma \epsilon}{(1 + \gamma) \psi} + \frac{1 + \frac1\gamma}{1 + \gamma} - \frac1\gamma}\sigma_a\\
        &= \frac{\gamma \epsilon \sigma_a}{(1 + \gamma) \psi} > 0
    \end{align}
    whence \cref{eq:balance} is satisfied for the pair $(\set{x}, \set{b})$. It follows that
    \begin{equation}
        R_\psi(G) \ge \frac{\sigma_a + \sigma_b}{\sigma_V} = \frac{\pa{1 + \frac1\gamma - \frac\epsilon\psi}\sigma_a}{\sigma_a/\psi} = \pa{1 + \frac1\gamma}\psi - \epsilon
    \end{equation}
    as desired.
\end{proof}

\section{Local Security}\label{sec:local}

The bound in \cref{thm:slack} on the worst-possible stake loss from cascading attacks is reassuring from a global perspective, but less so from the perspective of one or a small number of services who would like an assurance that they will not be among those affected by such attacks. 
Beginning with this section, 
we focus on a specific coalition of services $C \subseteq S$ that seeks to insulate their shared security $\Gamma(C)$ against shocks and resulting cascading attacks that may come about due to the decisions of other services and validators. Formally, we denote by 
\begin{equation}
    \Gamma(C) := \set{v \in V \mid N_G\set{v} \subseteq C}
\end{equation}
the set of validators that exclusively provide security for services in $C$. 

\paragraph{Worst-case stake loss (local version).}
As before, we first suppose that an initial shock affects the restaking graph. Whereas previously, we considered shocks for which the total stake in the shock was bounded, we now consider shocks for which the total stake that impacts the exclusive security of some coalition of services $C$ (i.e., stake that secures services from $C$ and only services from $C$) is bounded. Formally, for any coalition of services $C \subseteq S$ within some restaking graph $G$, we let

\begin{equation}
    \mathbb D_\psi(C, G) := \set{D \subseteq V \mid \frac{\sigma_{D \cap \Gamma(C)}}{\sigma_{\Gamma(C)}} \le \psi}
\end{equation}
denote the set of all validator coalitions that provide at most $\psi$ stake to the aggregate security of the coalition of services $C$. Notice that shocks $D \in \mathbb D_\psi(C, G)$ may have much more total stake $\sigma_D$ than a $\psi$ fraction of the graph. We are instead only guaranteed that the impact of the shock on stake that is being used exclusively for members in $C$ is small. We are now interested in the potential cascading losses that can affect the stake that is exclusively utilized by the coalition $C$ after a shock occurs that destroys at most $\psi$ stake from the aggregate security of $C$. Formally, we study the quantity

\begin{equation}
    R_\psi(C, G) := \underbrace{\psi}_{\text{Initial Shock}} + \max_{D \in \mathbb D_\psi(C, G)} \max_{(A_1, B_1), \dots, (A_T, B_T) \in \mc C(G \searrow D)} \underbrace{\frac{\sigma_{\bigcup_{t = 1}^T B_t \cap \Gamma(C)}}{\sigma_{\Gamma(C)}}}_{\text{Stake lost from cascading attacks}}
\end{equation}

\paragraph{Local security conditions.} We seek sufficient conditions on a restaking graph $G$ that guarantee a nontrivial upper bound on $R_\psi(C, G)$. 
Ideally, the sufficient condition
%In particular, a sufficient condition that 
would depend only the neighborhood/choices of the coalition $C$
%as 
%information about profits from corruption $\pi_s$ may only be known for the coalition, and the %coalition may want a condition that they can guarantee to be true 
to provide a guarantee that holds regardless of the choices made by other validators and services (i.e., the services of $C$ can attempt to ``control their own destiny'' by ensuring that the locally defined sufficient condition holds). Formally, given a restaking graph $G = (S, V, E, \pi, \sigma, \alpha)$ and a coalition of services $C \subseteq S$, we call a restaking graph $G' = (S', V', E', \pi', \sigma', \alpha')$ a \textit{$C$-local variant} of $G$ if $C$ cannot distinguish $G'$ from $G$ on the basis of local information: $C \subseteq S'$, $N_GC = N_{G'}C$, and
\begin{align}
    (\pi_s, \alpha_s) &= (\pi'_s, \alpha'_s) && \forall s \in C\\
    \pa{\sigma_v, N_G\set{v}} &= \pa{\sigma'_v, N_{G'}\set{v}} && \forall v \in N_GC
\end{align}
We then define a \textit{local security condition} $f: (C, G) \mapsto \set{0,1}$ to be a Boolean function that takes as input a restaking graph $G$ and a coalition of services $C \subseteq S$ such that $f(C, G)$ must be equal to $f(C, G')$ for all $C$-local variants $G'$. The intuition behind this definition is that the condition should only depend on service-level information (e.g. profits from corruption, security thresholds) for services in the coalition, and validator-level information for validators in $N_GC$. 

\paragraph{Local security impossibility.} Unfortunately, without further restrictions on the attacks under consideration (like those defined later in this section), it is impossible to construct any nontrivial local security condition that yields any nontrivial upper bound on $R_\psi(C, G)$.

\begin{theorem}\label{thm:localimpossible}
    For any local security condition $f$, any secure restaking graph $G$ and coalition of services $C \subseteq S$ such that $f(C, G) = 1$, there exists a secure $C$-local variant $G'$ of $G$ such that $R_0(C, G') = 1$. 
\end{theorem}

\begin{proof}
    Take any $f$, $C$, and secure $G$ such that $f(C, G) = 1$. Define 
    \begin{equation}
        \Delta := \sigma_{N_GC} - \pi_C
    \end{equation}
    to be the total overcollateralization of $C$ in aggregate. Next, we define $G'$ to be an augmented version of $G$, where we add a new service $s^*$ that has a profit from corruption $\pi_{s^*} = \Delta + 2\epsilon$ where $\epsilon > 0$. We further add $2$ validators $a$ and $b$ to the graph who are adjacent only to $s^*$. We let $\sigma_a := \Delta + \epsilon$ and $\sigma_b := \epsilon$. As $\sigma_a + \sigma_b \ge \pi_{s^*}$, the graph $G'$ must be secure as $G$ was secure. Next, notice that as the validator $a$ is not path-connected to $C$, $G'$ must be a $C$-local variant of $G$. However, by construction, the attack $\pa{C \cup \set{s^*}, N_GC \cup \set{b}}$ is valid on the graph $G' \searrow \set{a}$. It follows that $R_0(C, G') = 1$.
\end{proof}

\paragraph{Stable attacks.} While the above impossibility appears to be quite strong, it is somewhat contrived. At the heart of the impossibility is that under the definition of a valid attack (i.e. \cref{eq:balance,eq:threat}), not every validator must be productive in carrying out the attack. There may be a subset of validators in the attack that can yield more net profit than the full coalition. In what follows, we show that if we assume that malicious validator coalitions will choose to add others to their ranks only if it is profitable for them in net to do so, then a local security condition with guarantees similar to those in \cref{thm:slack} does indeed exist. Formally, for $A \subseteq S$ and $B \subseteq V$, we say that an attack $(A, B)$ is \textit{stable} if it is valid (i.e. \cref{eq:balance,eq:threat} hold), and for all $A' \subseteq A$ and $B' \subseteq B$ such that $(A', B')$ is valid, 
\begin{equation}\label{eq:stable}
    \sigma_{B \setminus B'} < \pi_{A \setminus A'} 
\end{equation}
Intuitively, if~\eqref{eq:stable} did not hold, then the validators of $B'$ would be better off ditching those in $B \setminus B'$ and attacking only the services in~$A'$.

We say that a disjoint sequence $(A_1, B_1), \dots, (A_T, B_T) \in 2^S \times 2^V$ is a cascade of stable attacks on a restaking graph $G = (S, V, E, \pi, \sigma, \alpha)$  if for each $t \in [T]$, $(A_t, B_t)$ is a stable attack on $G \searrow \bigcup_{i = 1}^{t-1} B_i$. We denote by $\mc S(G)$ the set of all such sequences of cascading stable attacks. In light of Theorem \ref{thm:localimpossible}, we redefine our notion of worst-case stake loss using stable attacks:
\begin{equation}
    R_\psi(C, G) := \psi + \max_{D \in \mathbb D_\psi(C, G)} \max_{(A_1, B_1), \dots, (A_T, B_T) \in \mc S(G \searrow D)} \frac{\sigma_{\bigcup_t B_t \cap \Gamma(C)}}{\sigma_{\Gamma(C)}}
\end{equation}
Unlike valid attacks, unions of sequences of stable cascading attacks need not be stable (which will complicate the proof of \cref{thm:local} in the next section).

\begin{claim}
    There exists a restaking graph $G$ and a sequence $(A_1, B_1), (A_2, B_2) \in \mc S(G)$ such that $\pa{A_1 \cup A_2, B_1 \cup B_2} \not \in \mc S(G)$.
\end{claim}

\begin{proof}
        Consider restaking graph where there are two services $S = \set{x, y}$ and two validators $V = \set{a,b}$ where both validators are restaking in both services. Furthermore, $\pi_x = \pi_y = 2$, $\alpha_x = \alpha_y = \frac12$, and $\sigma_a = \sigma_b = 1$. In this case, the sequence $\pa{\set{x}, \set{a}}, \pa{\set{y}, \set{b}} \in \mc S(G)$. However, $\pa{\set{x, y}, \set{a, b}} \not \in \mc S(G)$ because the attack $\pa{\set{x, y}, \set{a}}$ is a valid attack.
\end{proof}

\section{A Local Security Condition for Stable Attacks}\label{sec:local_ub}

In this section, we give a family of local security conditions that yield guarantees on the local worst-case stake loss $R_\psi(C, G)$ that resemble our result from Theorem \ref{thm:slack} for global security. In Theorems \ref{thm:slack} and \ref{thm:noslack}, we showed that security with $\gamma$-slack was necessary and sufficient in order to obtain a $\pa{1 + \frac1\gamma}\psi$ upper-bound on $R_\psi(G)$. In other words, it was both necessary and sufficient to ensure that all attacking coalitions $(A, B)$ were overcollateralized multiplicatively by a factor of $(1 + \gamma)$. Our condition for local security is similar: we must ensure that certain \textit{attack headers} $(X, Y)$ are overcollateralized multiplicatively by a factor of $(1 + \gamma)$.

\paragraph{Attack header.} Formally, we say that $(X, Y)$ is an {\em attack header}, for $X \subseteq S$ and $Y \subseteq \Gamma(X)$, if there exists a set of validators $B \subseteq V$ satisfying
\begin{equation}
    B \cap \Gamma(X) = \emptyset
\end{equation}
such that $(X, B \cup Y)$ is an attacking coalition. In other words, $(X, Y)$ is an attack header if $Y$ can be appended to a collection of validators $B$ that may be slashed without attacking services in $X$ to form an attacking coalition that attacks the services $X$. An attacking coalition is a special case of an attack header (in which~$B$ can be taken as $\emptyset$).

\begin{theorem}\label{thm:local}
    Let $G = (S, V, E, \pi, \sigma, \alpha)$ be a restaking graph and $C \subseteq S$ be a coalition of services. If, for all attack headers $(X, Y)$ where $X \subseteq C$, 
    \begin{equation}\label{eq:stubslack}
        (1 + \gamma) \pi_X \le \sigma_Y
    \end{equation}
    then $R_\psi(C, G) < (1+\frac{1}{\gamma})\psi$. Furthermore, the Boolean function that checks whether \cref{eq:stubslack} holds for all attack headers is a local security condition.
\end{theorem}

\begin{proof}
    Let $D \in \mathbb D_\psi(C, G)$ and $(A_1, B_1), \dots, (A_T, B_T) \in \mc S(G \searrow D)$ be arbitrary. For each $t \in [T]$, define
    \begin{equation}
        L_t := B_t \cap \Gamma(C)
    \end{equation}
    to be the set of all validators exclusively securing $C$ that were lost in the $t\textsuperscript{th}$ attack. Next, define
    \begin{equation}
        A'_t := \set{s \in A_t \mid \sigma_{B_t \setminus L_t} \ge \alpha_s \cdot \sigma_{N_G \set{s} \setminus \pa{\bigcup_{i = 1}^{t - 1} B_i \cup D}}}
    \end{equation}
    to be the maximal set of services such that $(A'_t, B_t \setminus L_t)$ is an attacking coalition on $G \searrow \pa{\bigcup_{i = 1}^{t-1} B_i \cup D}$. Notice also that because $(A_t, B_t)$ is a stable attack on $G \searrow \pa{\bigcup_{i = 1}^{t-1} B_i \cup D}$, we must have that $A_t \setminus A'_t$ is nonempty, and furthermore must be a subset of $C$. 
    
    \begin{claim} \label{claim:stubalance}
        \begin{equation} \sigma_{\bigcup_t B_t \cap \Gamma(C)} < \pi_{\bigcup_t A_t \setminus A'_t} 
        \end{equation}
    \end{claim}
    \begin{proof}
    From stability, we have that
    \begin{equation}
        \sigma_{L_t} = \sigma_{B_t \setminus \pa{B_t \setminus L_t}}  < \pi_{A_t \setminus A'_t} 
    \end{equation}
    To see why this holds, notice that there are two cases. If the attacking coalition $(A'_t, B_t \setminus L_t)$ is a valid attack, then the inequality follows directly from the stability definition. If instead $(A'_t, B_t \setminus L_t)$ is not a valid attack despite being an attacking coalition, the inequality must still hold because the original attack~$(A_t,B_t)$ is valid and therefore satisfies \cref{eq:balance}. Iterating over $t$, we find that
    \begin{equation}
        \sigma_{\bigcup_t B_t \cap \Gamma(C)}  = \sum_{t = 1}^T \sigma_{L_t} < \sum_{t = 1}^T \pi_{A_t \setminus A'_t} = \pi_{\bigcup_t A_t \setminus A'_t} 
    \end{equation}
    \end{proof}

    \begin{claim} \label{claim:istub}
        $\pa{\bigcup_t \pa{A_t \setminus A'_t}, \pa{\bigcup_t B_t \cup D} \cap \Gamma(C)}$ is an attack header on $G$, and $\bigcup_t \pa{A_t \setminus A'_t} \subseteq C$.
    \end{claim}
    \begin{proof}
    Because each $A_t \setminus A'_t \subseteq C$, we must also have that $\bigcup_t A_t \setminus A'_t \subseteq C$ as well. Next, by applying \cref{cor:combine} noting the disjointness of the $(A_t, B_t)$, we have that
    \begin{equation}
        \pa{\bigcup_t A_t, \bigcup_t B_t} = \pa{\bigcup_t \pa{A_t \setminus A'_t} \cup \bigcup_t A'_t, \quad \bigcup_t B_t}
    \end{equation}
    must be an attacking coalition on $G \searrow D$. In particular, we must have that $\pa{\bigcup_t \pa{A_t \setminus A'_t}, \bigcup_t B_t}$
    is an attacking coalition on $G \searrow D$, whence by Lemma \ref{cor:combine} we have that $\pa{\bigcup_t \pa{A_t \setminus A'_t}, \bigcup_t B_t \cup D}$
    is an attacking coalition on $G$. Rewriting the above as 
    \begin{equation}
        \pa{\bigcup_t \pa{A_t \setminus A'_t}, \quad \bra{\pa{\bigcup_t B_t \cup D} \cap \Gamma(C)} \cup \bra{\pa{\bigcup_t B_t \cup D} \setminus \Gamma(C)}}
    \end{equation}
    and noting by construction that
    \begin{equation}
        \bra{\pa{\bigcup_t B_t \cup D} \setminus \Gamma(C)} \cap \Gamma(C) = \emptyset,
    \end{equation}
    we find that $\pa{\bigcup_t \pa{A_t \setminus A'_t}, \pa{\bigcup_t B_t \cup D} \cap \Gamma(C)}$ is an attack header on $G$ as desired.
    \end{proof}
    Putting these claims together yields the desired result. By Claim \ref{claim:istub} and \cref{eq:stubslack}, we find that
    \begin{equation}
        (1+\gamma) \pi_{\bigcup_t A_t \setminus A'_t}  \le \sigma_{\pa{\bigcup_t B_t \cup D} \cap \Gamma(C)}  \le \sigma_{\bigcup_t B_t \cap \Gamma(C)} + \psi \cdot \sigma_{\Gamma(C)}
    \end{equation}
    Adding in Claim \ref{claim:stubalance}, we find that
    \begin{align}
        (1+\gamma)\sigma_{\bigcup_t B_t \cap \Gamma(C)} &< \sigma_{\bigcup_t B_t \cap \Gamma(C)} + \psi \cdot \sigma_{\Gamma(C)}\\
        &\implies \frac{\sigma_{\bigcup_t B_t \cap \Gamma(C)}}{\sigma_{\Gamma(C)}} < \frac{\psi}{\gamma}\\
        &\implies R_\psi(C, G) < \pa{1 + \frac1\gamma}\psi
    \end{align}
    as desired. To see that the Boolean function that checks whether \cref{eq:stubslack} holds for all attack headers is a local security condition, observe that it suffices to check that for all $X \subseteq C$ and $Y \subseteq N_GC \cap \Gamma(C) = \Gamma(C)$ such that $\pa{X, Y \cup N_GC \setminus \Gamma(C)}$ is an attacking coalition, \cref{eq:stubslack} holds. Thus, for any restaking graph $G$ and $C$-local variant $G'$, this function will evaluate to the same output.
\end{proof}

The condition in~\eqref{eq:stubslack} can be checked via enumeration, although the time required to do so grows exponentially in the number of services in~$C$ and the number of validators that contribute security exclusively to services in~$C$. 
In the event that this is a prohibitive amount of computation,
the same guarantee holds under a stronger, easily checked local analog of the EigenLayer sufficient condition~\eqref{eq:EL} which, in effect, treats as malicious all validators that contribute security to any services outside of $C$.

\begin{corollary}\label{cor:local}
    For any restaking graph $G = (S,V,E,\pi,\sigma,\alpha)$ and coalition of services $C \subseteq S$, satisfaction of the local condition
    \begin{equation}\label{eq:ELlocal}
        \sum_{s \in N_G\set{v}} \frac{\sigma_v}{\sigma_{N_G\pa{\set{s}}}} \cdot \frac{(1 + \gamma) \pi_s}{\alpha'_s} \le \sigma_v 
    \end{equation}
    for all $v \in N_GC$ with $N_G\set{v} \subseteq C$ guarantees that $R_\psi(C, G) < \pa{1 + \frac1\gamma} \psi$, where for each $s \in C$,
    \begin{equation}
        \alpha'_s := \alpha_s - \frac{\sigma_{N_G\set{s} \setminus \Gamma(C)}}{\sigma_{N_G\set{s}}}
    \end{equation}
\end{corollary}
\begin{proof}
    This follows from \cref{thm:local} and \cref{claim:EL}. Notice that to guarantee that the condition (\ref{eq:stubslack}) from \cref{thm:local} holds for all attack headers $(X,Y)$, it suffices to guarantee that no valid attacks exist, when (i) all profits from corruption for services in $C$ are inflated by a multiplicative factor of $(1+\gamma)$, and (ii) the fraction of stake required to corrupt a given service is offset by the fraction of stake in that service that is also restaking for services that do not belong to $C$. As \cref{claim:EL} provides a sufficient condition to guarantee the existence of no valid attacks, \cref{eq:ELlocal} will guarantee that all attack headers are overcollateralized by a multiplicative factor of $(1 + \gamma)$.
\end{proof}

\section{Lower Bounds for Local Security}\label{sec:local_lb}

We next show senses in which \cref{thm:local} (and more strongly, \cref{cor:local}) is tight.

\begin{corollary}\label{cor:localtight}
For any $\psi, \gamma, \epsilon > 0$ such that
    \begin{equation}
        0 \le \pa{1 + \frac1\gamma}\psi - \epsilon \le 1,
    \end{equation}
    there exists a restaking graph $G$ that satisfies the condition (\ref{eq:ELlocal}) from \cref{cor:local} but has $R_\psi(C, G) \ge \pa{1 + \frac1\gamma}\psi - \epsilon$.
\end{corollary}
\begin{proof}
    Notice that if we take $C = S$, the condition (\ref{eq:ELlocal}) from \cref{cor:local} is identical to the condition (\ref{eq:ELslack}) from \cref{cor:ELslack}. Furthermore, $R_\psi(S, G)$ is the same as $R_\psi(G)$ except for the fact that $R_\psi(S, G)$ considers only stable attacks. Repeating the argument from \cref{thm:noslack}, and noting that the attack given in the proof of that result is stable, we obtain the desired result. 
\end{proof}

\begin{theorem}\label{thm:globalinsufficient}
    For any $\gamma > 0$, there exists a graph restaking graph $G = (S, V, E, \pi, \sigma, \alpha)$ that satisfies (\ref{eq:ELslack}) from \cref{cor:ELslack}, but there exists a $C \subseteq S$ such that $R_0(C, G) = 1$.
\end{theorem}
\begin{proof}
    Let there be three services $S = \set{x, y , z}$, three validators $V = \set{a, b, c}$, and edges as follows:
    \begin{equation}
        E := \set{(x, a), (x, b), (y, b), (y, c), (z, c), (z, a)}.
    \end{equation}
    Next, let $\alpha_s = 1$ for all $s \in S$, and let $\pi_x = \pi_y = \pi_z =: \pi > 0$ be an arbitrary positive constant. Next pick $\sigma_a$ such that
    \begin{equation}
        \sigma_a < 2\pi,
    \end{equation}
    let $\sigma_b := 2(1 + \gamma)\pi$, and let $\sigma_c := 2(1 + \gamma)\pi$. This graph satisfies (\ref{eq:ELslack}) from \cref{cor:ELslack} as
    \begin{align}
        \frac{\sigma_a}{\sigma_a + \sigma_b} \cdot (1+\gamma) \pi + \frac{\sigma_a}{\sigma_a + \sigma_c} \cdot (1 + \gamma) \pi < (1 + \gamma)\pi\pa{\frac{1}{\sigma_b} + \frac{1}{\sigma_c}}\sigma_a = \sigma_a\\
        \frac{\sigma_b}{\sigma_a + \sigma_b} \cdot (1+\gamma) \pi + \frac{\sigma_b}{\sigma_b + \sigma_c} \cdot (1 + \gamma) \pi < \frac32 (1 + \gamma)\pi < \sigma_b\\
        \frac{\sigma_c}{\sigma_a + \sigma_c} \cdot (1+\gamma) \pi + \frac{\sigma_c}{\sigma_b + \sigma_c} \cdot (1 + \gamma) \pi < \frac32 (1 + \gamma)\pi < \sigma_c
    \end{align}
    Next, let $C := \set{x, z}$, and observe that the shock $D = \set{b, c}$ is an element of $\mathbb D_0(C, G)$. However, because $\sigma_a < \pi_x + \pi_z = 2\pi$, we must have that $(\set{x, z}, \set{a})$ is a stable attack on $G \searrow D$. As $\set{a} = \Gamma(C)$, it follows that $R_0(C, G) = 1$.
\end{proof}

\section{Long Cascades}\label{sec:length}

\paragraph{Cascade Structure.} Although \cref{cor:combine} shows that all long cascades can be made into a short, one-step attack, it is arguably more dangerous if large attacks can be made through long cascades of small attacks that each require less coordination. In this section, we show how adding $\gamma$-slack also enables us to upper-bound the length of a cascade in the worst case. Our results depend on what we call the \textit{reference depth} of a cascading sequence of attacks. Formally, for a given restaking graph $G$ and coalition of validators $B_0$, we say that a valid cascade of attacks $(A_1, B_1), \dots, (A_T, B_T)$ on $\mc C(G\searrow {B_0})$ has reference depth $k$ if

\begin{equation}
    k = \max \set{i \in [T] \mid \exists  t \in [T] \text{ s.t. } N_G(A_t) \cap B_{t - i} \ne \emptyset}
\end{equation}

In other words, a cascading sequence has reference depth $k$ if the services attacked in a given time step are not affected by validators that were slashed more than $k$ steps previous. 

\begin{theorem} \label{thm:length}
    Suppose that a restaking graph $G = (S, V, E, \pi, \sigma, \alpha)$ is secure with $\gamma$-slack for some $\gamma > 0$. Let $\epsilon = \min_{v \in V} \sigma_v$ denote the minimum stake held by a validator. Then, for any $\psi > 0$, $B_0 \in \mathbb D_\psi(G)$, and $(A_1, B_1), \dots, (A_T, B_T) \in \mc C(G \searrow {B_0})$ with reference depth $k$, 
    \begin{equation}
        T < k\pa{1 + \log_{1 + \gamma}\frac{\psi \cdot \sigma_V}{\epsilon \gamma}}
    \end{equation}
\end{theorem}

\begin{proof}
    Because $(A_1, B_1), \dots, (A_T, B_T) \in \mc C(G \searrow {B_0})$ has reference depth $k$, it follows that if we define
    \begin{align}
        A'_i &:= \bigcup_{t = k(i-1) + 1}^{ki} A_t\\
        B'_i &:= \bigcup_{t = k(i-1) + 1}^{ki} B_t\\
        B'_0 &:= B_0
    \end{align}
    for $i \in \set{1, \dots, \ceil{T/k}}$, we must have that for all $j < i - 1$,
    \begin{equation}\label{eq:canonical}
        N_GA'_i \cap B'_j = \emptyset
    \end{equation}
    Furthermore, by \cref{cor:combine}, we must have that the sequence $(A'_i, B'_i) \in \mc C(G\searrow B_0)$. Applying \cref{cor:combine} again, we further have that for every $i$, $\pa{\bigcup_{j = i + 1}^{\ceil{T/k}} A'_i, \bigcup_{j = i + 1}^{\ceil{T/k}} B'_i}$ is a valid attack on $G \searrow \bigcup_{j = 0}^{i} B'_i$. It follows by \cref{eq:balance} that for any $i \in \set{0, \dots,\ceil{T/k} - 1}$,
    \begin{align}\label{eq:iterbalance}
        \pi_{\bigcup_{j = i + 1}^{\ceil{T/k}} A'_j} > \sigma_{\bigcup_{j = i + 1}^{\ceil{T/k}} B'_j} = \sum_{j = i + 1}^{\ceil{T/k}} \sigma_{B'_j}  
    \end{align}
    Next, noting that 
    \begin{equation}
        N_G \bigcup_{j = i+1}^{\ceil{T/k}} A_i' \cap \bigcup_{j = 1}^{i - 1} B_i' = \emptyset
    \end{equation}
    by \cref{eq:canonical}, it follows that indeed $\pa{\bigcup_{j = i + 1}^{\ceil{T/k}} A'_i, \bigcup_{j = i + 1}^{\ceil{T/k}} B'_i}$ is also a valid attack on $G \searrow B_i'$. Applying \cref{lem:acombine}, we then find that $\pa{\bigcup_{j = i + 1}^{\ceil{T/k}} A'_i, B'_i \cup \bigcup_{j = i + 1}^{\ceil{T/k}} B'_i}$ is an attacking coalition on the original graph $G$. Because $G$ is secure with $\gamma$-slack, \cref{eq:slack} must now hold on this pair, whence for all $i \in \set{0, \dots, \ceil{T/k} - 1}$,
    \begin{equation}
        (1+\gamma) \pi_{\bigcup_{j = i + 1}^{\ceil{T/k}} A'_j} \le \sigma_{\bigcup_{j = i}^{\ceil{T/k}} B'_j} = \sigma_{B'_i} + \sum_{j = i+1}^{\ceil{T/k}} \sigma_{B'_j}
    \end{equation}
    Putting this together with \cref{eq:iterbalance}, we have that for all $i \in \set{0, \dots, \ceil{T/k} - 1}$,
    \begin{align}
        &(1 + \gamma) \sum_{j = i+1}^{\ceil{T/k}} \sigma_{B'_j} <  \sigma_{B'_i} + \sum_{j = i+1}^{\ceil{T/k}} \sigma_{B'_j}\\
        &\implies \sigma_{B'_i}  > \gamma \sum_{j = i+1}^{\ceil{T/k}} \sigma_{B'_j} \label{eq:itergamma}
    \end{align}
    It can be shown inductively from \cref{eq:itergamma} that the sequence $\sigma_{B'_i} > X_i$, where we define the sequence $X_i$ by
    \begin{align}
        X_{\ceil{T/k}} &:= \epsilon\\
        X_i &:= \gamma \sum_{j = i+1}^{\ceil{T/k}} X_j && i \in \set{0, \dots, \ceil{T/k} - 1}
    \end{align}
    Solving, we find that
    \begin{equation}
        X_0 = \epsilon \gamma (1 + \gamma)^{\ceil{T/k} - 1}
    \end{equation}
    It then follows as desired that
    \begin{align}
        &\psi \cdot \sigma_V = \sigma_{B_0} > \epsilon \gamma(1+\gamma)^{T/k - 1} \implies T < k\pa{1 + \log_{1 + \gamma}\frac{\psi \cdot \sigma_V}{\epsilon \gamma}}
    \end{align}
\end{proof}

\subsection*{Acknowledgments}

We thank Tarun Chitra, Soubhik Deb, Sreeram Kannan, Mike Neuder, and
Mallesh Pai for comments on earlier drafts of this paper.  We thank
Soubhik and Sreeram in particular for emphasizing the importance of
local guarantees.

This research was supported in part by
NSF awards CCF-2006737 and CNS-2212745, and research awards from the
Briger Family Digital Finance Lab and the Center for Digital Finance
and Technologies.

\bibliographystyle{plain}
\bibliography{main}

\newpage
\appendix

\section{Lower Bounds under a Bounded Profit-Stake Ratio}

The example exhibited in \cref{thm:gn} requires that for the whole graph to be attacked due to an arbitrarily small shock $R_\epsilon(G) = 1$, there is a service whose profit from corruption is arbitrarily larger than the stake of some validator (i.e. $\pi_s/\sigma_v$ for a service $s$ and validator $v$). In this section, we show that equally strong lower bounds for both global and local security exist even when all of the services are ``small'', in the sense that an absolute constant bounds the ratio of the profit from corruption of any service and the stake of any validator. 

\begin{theorem}\label{thm:localtight2}
    For any $0 < \epsilon < 1$, there exists a secure restaking graph $G = (S, V, E, \pi, \sigma, \alpha)$ that meets the EigenLayer condition (\ref{eq:EL}),
    \begin{equation}\label{eq:bounded}
        \max_{(s, v) \in S \times V} \pi_s/\sigma_v = 2,
    \end{equation}
    but there exist proper subsets $C \subset S$ such that $R_0(C, G) = 1$. Furthermore, for any $\psi \ge \epsilon$, $R_{\psi}(S, G) = 1$.
\end{theorem}
\begin{proof}
    Pick any $n$ such that $6 \mid n$ and $\frac1n < \epsilon$. We let $V$ consist of $n$ validators each with stake $\sigma_v = 1$. Letting $N := n/6$, we let $S$ consist of $3N$ services each with $\pi_s := 2$ and $\alpha_s := 1$. The graph thus satisfies \cref{eq:bounded}. We now split $S$ into two disjoint sets $S_6$ and $S_3$, where $S_6$ contains $N$ of the services, and $S_3$ contains the remaining $2N$ services. We define the edges $E$ as follows. First, we draw edges such that the neighborhoods $\set{N_G\set{s}}_{s \in S_6}$ of services in $S_6$ evenly partition $V$ into groups of $6$. Next, we draw edges such that the neighborhoods $\set{N_G\set{s}}_{s \in S_3}$ evenly partition $V$ into groups of $3$, and make the entire graph $G$ a single connected component. An example of such a graph for $n = 30$ is shown in \cref{fig:gn}. The graph satisfies (\ref{eq:EL}) as each validator is adjacent to just one service from $S_6$ and one service from $S_3$, whence for any $v \in V$,
    \begin{equation}
        \sum_{s \in N_G \set{v}} \frac{\sigma_v}{\sigma_{N_G\set{s}}} \cdot \frac{\pi_s}{\alpha_s} = \frac{1}{6} \cdot 2 + \frac{1}{3} \cdot 2 = \sigma_v
    \end{equation}
    Thus, $G$ is secure by \cref{claim:EL}.
    \begin{claim}\label{clm:proper}
        For all attacking coalitions $(A, B)$ on $G$ such that $(A, B) \ne (S, V)$, $\sigma_B \ge \pi_A + 1$.
    \end{claim}
    \begin{proof}
        To see this, notice first that as $\alpha_s = 1$ for all $s \in S$, for any $A \subseteq S$, $(A, N_GA)$ is the unique stake-minimal attacking coalition on $G$ that attacks $A$. Writing $A_3 := A \cap S_3$ and $A_6 := A \cap S_6$, we must further have that
        \begin{equation}
            \sigma_{N_GA} = 3\abs{A_3} + \abs{N_GA_6 \setminus N_GA_3}
        \end{equation}
        and that
        \begin{equation}
            \pi_A = 2\abs{A} = 2\pa{\abs{A_3} + \abs{A_6}}
        \end{equation}
        It follows that
        \begin{align}\label{eq:bomb}
            \sigma_{N_GA} - \pi_A = \abs{A_3} - 2\abs{A_6} + \abs{N_GA_6 \setminus N_GA_3}
        \end{align}
        As all validators $v$ have $\sigma_v = 1$, it suffices to show that $\sigma_{N_GA} > \pi_A$. Suppose for contradiction that $\sigma_{N_GA} \le \pi_A$. We must then have by \cref{eq:bomb} that $\abs{A_3} \le 2 \abs{A_6}$ else $\sigma_{N_GA} > \pi_A$. Furthermore, $2 \abs{A_6} \le \abs{A_3}$ since $\abs{N_GA_6 \setminus N_GA_3} \ge 3\pa{2\abs{A_6} - \abs{A_3}}$. It follows that we must have $\abs{A_3} = 2\abs{A_6}$ and $\abs{N_GA_6 \setminus N_GA_3} = 0$. However, if both of these conditions occurred simultaneously, then $(A, N_GA)$ must be disconnected from the rest of the graph. As $G$ forms a single connected component, and $N_GA$ is a proper subset of $V$, it follows that $\sigma_{N_GA} > \pi_A$ as desired.
    \end{proof}
    
    Next, pick any $v \in V$, and consider the initial shock $D = \set{v}$. Notice that for any $C$ such that $v \not \in \Gamma(C)$, $D \in \mathbb D_0(C, G)$. Further, as $\epsilon > \frac1n$, $D \in \mathbb D_\psi(S, G)$ for any $\psi \ge \epsilon$. To finish the result, it suffices to show that $(S, V \setminus \set{v})$ is a stable attack on $G \searrow D$. To see this, note that it is clearly an attacking coalition, and 
    \begin{equation}
            \pi_S = n > n - 1 = \sigma_{V \setminus \set{v}}
    \end{equation}
    whence it is a valid attack on $G \searrow D$. To see that it is a stable attack, it suffices to show that there is no $A \subseteq S$ and proper $B \subset V\setminus \set{v}$ such that $(A, B)$ is a valid attack. It further suffices to consider proper $A \subset S$ as $V \setminus \set{v}$ is unique a subset of validators $B$ such that $(S, B)$ is an attacking coalition on $G \searrow D$. The result now follows from \cref{clm:proper} as for any such attacking coalitions $(A, B)$ on $G \searrow D$,
    \begin{equation}
        \sigma_{B \cup \set{v}} \ge \pi_A + 1 \implies \sigma_{B} \ge \pi_A 
    \end{equation}
    whence $(A, B)$ is not a valid attack on $G \searrow D$ as desired.
    \begin{figure}[t!]
        \centering
        \includegraphics[scale=0.3]{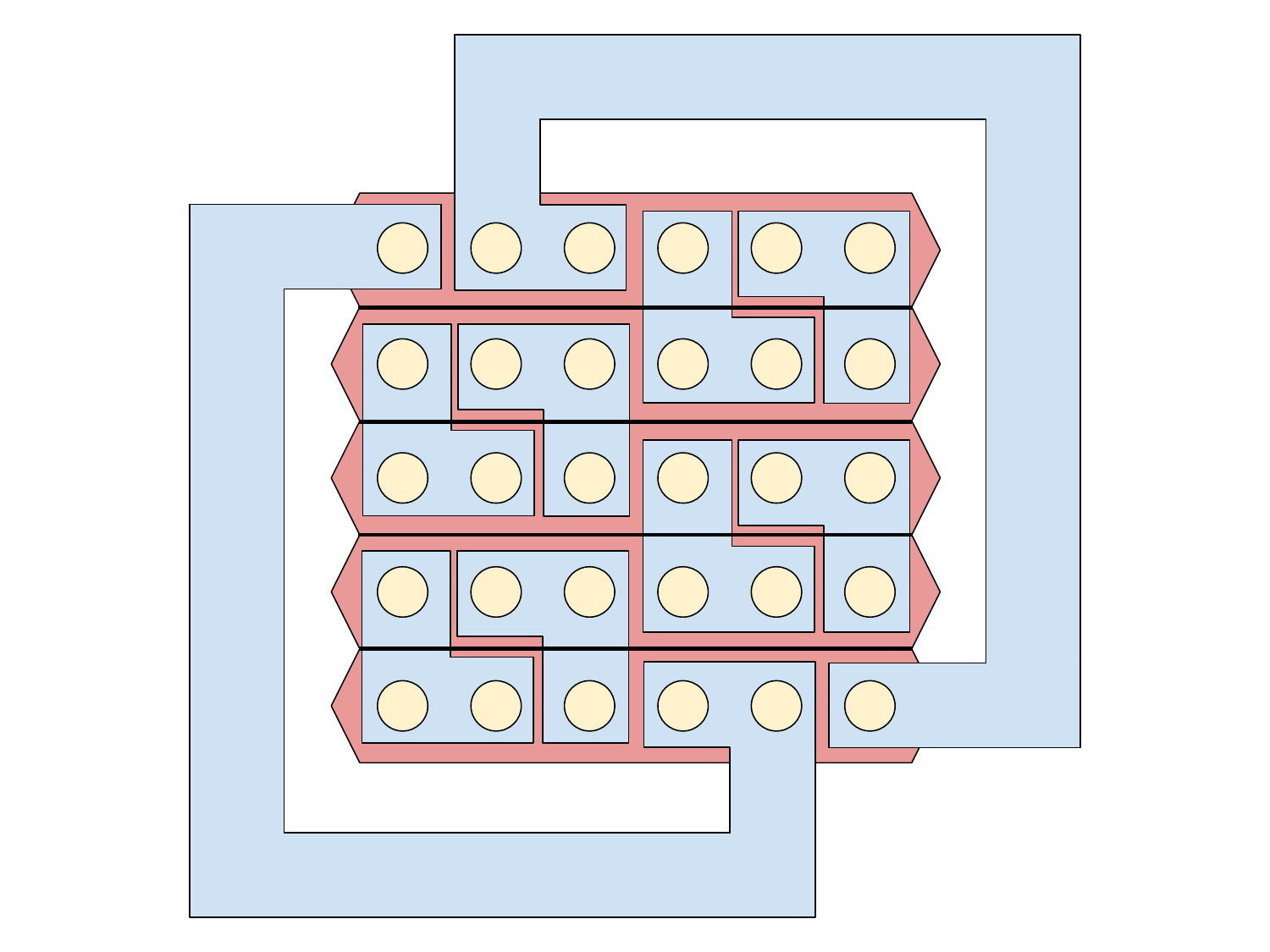}
        \caption{\textbf{The graph $G$ for $n = 30$.} Here, each circle denotes a validator. The red services constitute $S_6$, and the blue services constitute $S_3$. Observe that the graph forms a single connected component.}
        \label{fig:gn}
    \end{figure}
\end{proof}
\end{document}